\documentclass[singlecolumn,aps,nofootinbib,superscriptaddress,tightenlines,10pt]{revtex4}
\usepackage{amsmath}
\usepackage{amsthm}
\usepackage{amsfonts}
\usepackage{color}
\usepackage{mathtools}
\usepackage{overpic}

\usepackage[ colorlinks = true,
             linkcolor = blue,
             urlcolor  = blue,
             citecolor = red,
             anchorcolor = green,
]{hyperref}

\newtheorem{lemma}{Lemma}
\newtheorem{prop}[lemma]{Proposition}
\newtheorem{theorem}[lemma]{Theorem}
\newtheorem{definition}[lemma]{Definition}

\DeclareMathOperator{\tr}{tr}

\newcommand{\cS}{\mathcal{S}}
\newcommand{\cL}{\mathcal{L}}

\newcommand{\cH}{\mathcal{H}}
\newcommand{\cE}{\mathcal{E}}
\newcommand{\eps}{\varepsilon}

\begin{document}

\title{The Fidelity of Recovery is Multiplicative}

\author{Mario Berta}
\affiliation{Institute for Quantum Information and Matter, Caltech, Pasadena, California 91125, USA}
\author{Marco Tomamichel}
\affiliation{School of Physics, University of Sydney, Sydney, NSW 2006, Australia}

\begin{abstract}
Fawzi and Renner [Commun.~Math.~Phys.~340(2):575, 2015] recently established a lower bound on the conditional quantum mutual information (CQMI) of tripartite quantum states $\rho_{ABC}$ in terms of the fidelity of recovery (FoR), i.e.\ the maximal fidelity of the state $\rho_{ABC}$ with a state reconstructed from its marginal $\rho_{BC}$ by acting only on the $C$ system. The FoR measures quantum correlations by the local recoverability of global states and has many properties similar to the CQMI. Here we generalize the FoR and show that the resulting measure is multiplicative by utilizing semi-definite programming duality. This allows us to simplify an operational proof by Brand{\~a}o {\it et al.}~[Phys.~Rev.~Lett.~115(5):050501, 2015] of the above-mentioned lower bound that is based on quantum state redistribution. In particular, in contrast to the previous approaches, our proof does not rely on de Finetti reductions.
\end{abstract}

\maketitle


\section{Introduction}

The conditional quantum mutual information (CQMI) of a tripartite quantum state $\rho_{ABC}$ is defined as
\begin{align}
I(A\!:\!B|C)_{\rho} &:= H(AC)_{\rho} + H(BC)_{\rho}- H(ABC)_{\rho} - H(C)_{\rho}\,,
\end{align}
where $H(X)_{\rho} := -\tr[ \rho_{X} \log \rho_{X} ]$ denotes the von Neumann entropy. The CQMI is a measure for the correlations between $A$ and $B$ from the perspective of $C$ and has an operational interpretation as the optimal quantum communication cost in quantum state redistribution~\cite{devetak08,devetak09}. Apart from that the CQMI has found manifold applications in information theory~\cite{christandl04,brandao11}, physics~\cite{poulin11,kimthesis,kim14,kim14b}, as well as computer science~\cite{braverman12,brandao13b,brandao13c,touchette14}. 

A celebrated result known as strong subadditivity of entropy states that the CQMI is always non-negative~\cite{lieb73},
\begin{align}\label{eq:SSA}
I(A\!:\!B|C)_{\rho}\geq0\,.
\end{align}
Following a line of works (see~\cite{kimthesis,zhang13,bertawilde14,liwinter14} and references therein), Fawzi and Renner have shown in a recent breakthrough result that the lower bound~\eqref{eq:SSA} can be improved to~\cite{fawzirenner14},
\begin{align}\label{eq:fawzirenner}
I(A\!:\!B|C)_{\rho}\geq-\log F(A;B|C)_{\rho}\,,
\end{align}
where we have the fidelity of recovery (FoR),
\begin{align}\label{eq:fidrec}
F(A;B|C)_{\rho} := \max_{\Gamma_{C\to AC} } F\big(\rho_{ABC}, \left(\mathcal{I}_{B}\otimes\Gamma_{C\to AC}\right)(\rho_{BC}) \big) \,. 
\end{align}
Here, $\mathcal{I}_{B}$ denotes the identity channel on $B$, the supremum is taken over all quantum channels (completely positive and trace-preserving (CPTP) maps) from $C$ to $AC$, and we use $F(\cdot,\cdot)$ to denote Uhlmann's fidelity~\cite{uhlmann85}: $F(\rho, \sigma) := (\tr \big|\sqrt{\rho}\sqrt{\sigma}\big|)^2$. The FoR was defined and explored in detail by Seshadreesan and Wilde in~\cite{seshadreesan14}, where they show that it has similar properties as the CQMI. For example for pure states $\sigma_{ABCD}$ the monogamy of entanglement implies the duality
\begin{align}
I(A;B|C)_{\sigma}=I(A;B|D)_{\sigma}\quad\text{as well as}\quad F(A;B|C)_{\sigma}=F(A;B|D)_{\sigma}\,.
\end{align}
Similar to the squashed entanglement which is an entanglement measure based on the CQMI~\cite{christandl04}, the FoR then serves as the basis for a (pseudo) entanglement measure: the geometric squashed entanglement~\cite{seshadreesan14}. By its definition~\eqref{eq:fidrec} the FoR is also connected to the local recoverability of global quantum states, a promising concept for understanding topological order in condensed matter systems~\cite{kimthesis,kim14,kim14b}.

In this paper we investigate and generalize the FoR, and give an information theoretic proof of the lower bound in~\eqref{eq:fawzirenner}. For any two bipartite states $\rho_{AB}$ and $\sigma_{AC}$ (that may or may not have the same marginal on $A$) we define the generalized fidelity of recovery (FoR) as
\begin{align}\label{eq:def2}
F_{C\to B}(\rho_{AB}\|\sigma_{AC}) := \max_{\Gamma_{C\to B}} F\big( \rho_{AB} , \Gamma_{C\to B}(\sigma_{AC}) \big)\,.
\end{align}
We will drop the subscript $C \to B$ in the following when it is evident from the context on which systems the maps act. The original FoR as in~\eqref{eq:fidrec} is then simply given as $F_{C\to AC}(\rho_{ABC}\|\rho_{BC})=F(A;B|C)_{\rho}$. We note that the generalized form~\eqref{eq:def2} gives a way of comparing quantum information that lives on different dimensional systems. In particular, this allows to study quantum correlations independent of the system used to represent them.

Our results about the FoR are then as follows. We study a semi-definite programming (SDP) formulation of the FoR and find the dual minimization problem (Section~\ref{sec:sdp}).\footnote{SDPs are a powerful tool in quantum information theory with many applications (see~\cite{watrous-ln11} for an introduction).} Based on this dual SDP formulation, we establish our main technical result and show that the FoR is multiplicative for product states (Section~\ref{sec:mult}). Thus, we find that in particular for any two states $\rho_{ABC}$ and $\tau_{A'B'C'}$,
\begin{align}\label{eq:multi_main}
F(AA';BB'|CC')_{\rho\otimes\tau}= F(A;B|C)_{\rho} \cdot F(A';B'|C')_{\tau} \,.
\end{align}
This implies that there exists an optimal recovery map that has product structure as well. Additivity (multiplicativity) results are at the heart of quantum information theory, and using our finding~\eqref{eq:multi_main} we provide an operational proof of the Fawzi-Renner lower bound~\eqref{eq:fawzirenner} (Section~\ref{sec:fawzirenner}). This proof utilizes a connection between the fidelity of recovery and one-shot quantum state redistribution, using the ideas of Brand{\~a}o {\it et al.}~\cite{brandao14}.


\section{Fidelity of Recovery as an SDP}\label{sec:sdp}

In the following we denote the set of quantum states on a finite-dimensional Hilbert space $A$ by $\cS(A)$ and consequently use $\cS(ABC)$ to denote states on a tripartite quantum system $ABC$. We use subscripts to indicate on which Hilbert spaces an operator acts. The dimension of $A$ is denoted by $d_A$.

First, note that if $\rho_{AB}=\Phi_{AB}$ is a (normalized) maximally entangled state (MES) with $d_B=d_A$ then we get for the FoR~\eqref{eq:def2} by standard SDP duality the conditional min-entropy~\cite[Theorem 2]{koenig08},
\begin{align}\label{eq:min_mes}
F(\Phi_{AB}\| \sigma_{AC})=\frac{1}{d_A}2^{-H_{\min}(A|C)_{\sigma}}=\frac{1}{d_A}\min_{\omega_{C}\in\mathcal{S}(C)} \Big\| \sqrt{\omega_C^{-1}} \sigma_{AC} \sqrt{\omega_C^{-1}} \Big\|\,,
\end{align}
where $\|\cdot\|$ denotes the operator norm. In general we do not know much about the optimal $\omega_C$ in this expression except that for tensor-product states it is of tensor-product form as well. That is, the conditional min-entropy is additive~\cite{koenig08}. Second, if $\rho_{AB}=\psi_{AB}$ is a pure state with $d_B=d_A$ then we get by standard SDP duality~\cite[Remark 1]{koenig08},
\begin{align}\label{eq:min_pure}
F\big(\psi_{AB} \big\| \sigma_{AC}\big)= \min_{\omega_{C}\in\mathcal{S}(C)} \Big\| \sqrt{\psi_A \otimes\omega_C^{-1}}\, \sigma_{AC} \sqrt{\psi_A \otimes\omega_C^{-1}} \Big\|\,,
\end{align}
This quantity was also studied by Barnum and Knill in the context of quantum error correction~\cite{barnum02}. The fidelity of recovery can be formulated as an SDP in general.

\begin{lemma}\label{lem:alberti_recovery}
  Let $\rho_{AB} \in \cS(AB)$ and $\sigma_{AC} \in \cS(AC)$ and let $\sigma_{ACD}$ be a purification of $\sigma_{AC}$. Then, $F_{C \to B}(\rho_{AB}\|\sigma_{AC})$ as in~\eqref{eq:def2} is the solution to the following minimization problem:
\begin{align}\label{eq:sdp}
  \begin{array}{rl}
  \textnormal{minimize}:\ & \tr \big[\rho_{AB} R_{AB}^{-1} \big] \cdot \tr \big[\sigma_{AD} Q_{AD} \big] \\[5pt]
  \textnormal{subject to}:\ & R_{AB} > 0,\ Q_{AD} > 0 \, \\
  & Q_{AD} \otimes 1_B \geq R_{AB} \otimes 1_D \,.
  \end{array} 
\end{align}
\end{lemma}

The proof uses SDP duality following the footsteps of Watrous' lecture notes~\cite{watrous-ln11}. In particular Watrous discusses the dual SDP for the fidelity~\cite{watrous12} in the form of Alberti~\cite{alberti83},
\begin{align}\label{eq:alberti_original}
F(\rho,\sigma)=\min_{R>0}\tr \big[\rho R^{-1}\big] \cdot \tr \big[\sigma R\big]\,.
\end{align}
Our resulting dual program~\eqref{eq:sdp} can be thought of as the Alberti form of the fidelity of recovery (since it simplifies to~\eqref{eq:alberti_original} for trivial $B$ and $C$ systems).\footnote{Interestingly this Alberti form~\eqref{eq:sdp} does not directly simplify to~\eqref{eq:min_mes} and~\eqref{eq:min_pure} for the special case of MES and pure states.}

\begin{proof}[Proof of Lemma~\ref{lem:alberti_recovery}]
First note that using the purification $\sigma_{ACD}$, we can write
  \begin{align}
    \Gamma_{C\to B}(\sigma_{AC})= \tr_{D} \big[\Gamma_{C\to B}(\sigma_{ACD}) \big]= \tr_{D} \Big[\Gamma_{C \to B} \big( \sqrt{\sigma_{AD}} \Psi_{AD:C} \sqrt{\sigma_{AD}} \big) \Big]\,,
  \end{align}
  where we denote by $\Psi_{AD:C}$ the (unnormalized) maximally entangled state between $AD:C$ in the Schmidt decomposition of $\sigma_{ACD}$. Then, define the Choi-Jamio{\l}kowski state (unnormalized) of the map $\Gamma$ as
  \begin{align}
    \tau_{ABD} = \Gamma_{C \to B} \big( \Psi_{AD:C} \big) , \quad \tr_B[\tau_{ABD}] = 1_{AD}\,. \label{eq:choistate}
  \end{align}
  Hence, we can write
  \begin{align}
  \Gamma_{C\to B}(\sigma_{AC}) = \tr_D\big[\sqrt{\sigma_{AD}} \tau_{ABD} \sqrt{\sigma_{AD}}\big]\,,
  \end{align}
  and thus we can express the optimization problem for $\Gamma_{C \to B}$ in terms of the Choi-Jamio{\l}kowski state in~\eqref{eq:choistate}. On the other hand, every state $\tau_{ABD}$ satisfying~\eqref{eq:choistate} corresponds to a CPTP map $\Gamma_{C\to B}$. Hence, we can optimize over all Choi-Jamio{\l}kowski states of the form~\eqref{eq:choistate} instead of CPTP maps from $B$ to $C$. This leads to the following expression for the fidelity of recovery:
 \begin{align}\label{eq:choi_fidelity}
F(\rho_{AB}\| \sigma_{AC})= \max \Big\{  &F\Big(\rho_{AB}, \tr_D\big[\sqrt{\sigma_{AD}} \tau_{ABD} \sqrt{\sigma_{AD}}\big] \Big) :\;\tau_{ABD} \geq 0,\ \tr_{B}[\tau_{ABD}] = 1_{AD} \Big\}\,.
 \end{align}
 
 The primal problem above is obtained by considering an SDP for the root fidelity $\sqrt{F(\rho,\sigma)}=\|\sqrt{\rho}\sqrt{\sigma}\|_{1}$ as in~\cite[Section 2.1]{watrous12} and~\cite{killoranthesis}. The square root of the fidelity of recovery in~\eqref{eq:choi_fidelity} is then written as
\begin{align}
  \begin{array}{rl}
    \textnormal{maximize}:\ & \frac12 \tr\big[Z_{AB} + Z_{AB}^{\dag}\big] \\[5pt]
    \textnormal{subject to}:\ & \tau_{ABD} \geq 0,\ Z_{AB} \in \cL(AB), \\
    & \tr_B [\tau_{ABD}] = 1_{AD}, \\[1pt]
    & \begin{pmatrix} \rho_{AB} & Z_{AB} \\ Z_{AB}^{\dag} & \tr_D\big[\sqrt{\sigma_{AD}} \tau_{ABD} \sqrt{\sigma_{AD}}\big] \end{pmatrix} \geq 0 \,,
    \end{array} 
    \label{eq:primal}
\end{align}
where $\cL$ denotes the set of linear operators. In the next step we bring this program into standard form. We want to write the primal problem as a maximization over $X \geq 0$ of the functional $\tr[X A]$ subject to $\Phi(X) = B$. Hence, we set
 \begin{align}
     X = \left( \begin{matrix} X_{11} & Z_{AB} & \cdot \\ Z_{AB}^{\dag} & X_{22} & \cdot \\ \cdot & \cdot & \tau_{ABD} \end{matrix} \right) , \quad 
     A = \frac12 \begin{pmatrix} 0 & 1_{AB} & 0 \\ 1_{AB} & 0 & 0 \\ 0 & 0 & 0 \end{pmatrix}, \quad B = \begin{pmatrix} \rho_{AB} & 0 & 0 \\ 0 & 0 & 0 \\ 0 & 0 & 1_{AD} \end{pmatrix}\,,
\end{align}
as well as
\begin{align}
     \Phi(X)=\begin{pmatrix} X_{11} & 0 & 0 \\ 0 & X_{22} - \tr_D\big[\sqrt{\sigma_{AD}} \tau_{ABD} \sqrt{\sigma_{AD}}\big]  & 0 \\ 0 & 0 & \tr_{B} [\tau_{ABD}] \end{pmatrix}\,.
\end{align}     
The variables with the placeholder `$\cdot$' are of no interest to us.
  The dual SDP is a minimization over self-adjoint $Y$ of the functional $\tr[Y B]$ subject to $\Phi^{\dag}(Y) \geq A$. The dual variables and adjoint map can be determined to be
 \begin{align}
     Y = \left( \begin{matrix} L_{AB} & \cdot & \cdot \\ \cdot & R_{AB} & \cdot \\ \cdot & \cdot & Q_{AD} \end{matrix} \right)
 \end{align}
with
\begin{align}
\Phi(Y) =\begin{pmatrix}  L_{AB} & 0 & 0 \\ 0 & R_{AB} & 0 \\ 0 & 0 & -\sqrt{\sigma_{AD}} \big( R_{AB} \otimes 1_D \big) \sqrt{\sigma_{AD}} + 1_B \otimes Q_{AD}\end{pmatrix}\,.
\end{align}
This leads to the following dual problem:
\begin{align}\label{eq:dual_first}
  \begin{array}{rl}
  \textnormal{minimize}:\ &  \tr [\rho_{AB} L_{AB} ] + \tr [Q_{AD} ]\\[5pt]
  \textnormal{subject to}:\ & L_{AB}, R_{AB} \in \cH(AB),\ Q_{AD} \in \cH(AD), \, \\
  & Q_{AD}\otimes 1_B  \geq \sqrt{\sigma_{AD}} (R_{AB} \otimes 1_D) \sqrt{\sigma_{AD}} \,, \\[1pt]
  & \begin{pmatrix} L_{AB} & 0 \\ 0 & R_{AB} \end{pmatrix}\geq \frac12 \begin{pmatrix} 0 & 1_{AB} \\ 1_{AB} & 0 \end{pmatrix} \,,
  \end{array} 
\end{align}
where $\cH$ denotes the set of self-adjoint operators. The Slater condition (cf.~\cite{watrous-ln11}) for strong duality is satisfied. The program~\eqref{eq:dual_first} can be simplified further by the substitutions $L_{AB} \to \frac12 L_{AB}$, $R_{AB} \to \frac12 R_{AB}$ and $Q_{AD} \to \frac12 \sqrt{\sigma_{AD}} Q_{AD} \sqrt{\sigma_{AD}}$, leaving us with
\begin{align}
  \begin{array}{rl}
  \textnormal{minimize}:\ & \frac12 \tr \big[\rho_{AB} L_{AB} \big] + \frac12 \tr \big[\sigma_{AD} Q_{AD} \big]\\[5pt]
  \textnormal{subject to}:\ & L_{AB}, R_{AB} \in \cH(AB),\ Q_{AD} \in \cH(AD), \, \\
  & Q_{AD}\otimes 1_B \geq R_{AB} \otimes 1_D \,, \\[1pt]
  & \begin{pmatrix} L_{AB} & -1_{AB} \\ -1_{AB} & R_{AB} \end{pmatrix}\geq 0 \,.
  \end{array} 
\end{align}
Now we note that the above matrix inequality holds if and only if $L_{AB}, R_{AB} \geq 0$ and $R_{AB} \geq L_{AB}^{-1}$. Without loss of generality we can choose $R_{AB} = L_{AB}^{-1}$, and our problem simplifies to
\begin{align}
  \begin{array}{rl}
  \textnormal{minimize}:\ & \frac12 \tr \big[\rho_{AB} R_{AB}^{-1} \big] + \frac12 \tr \big[\sigma_{AD} Q_{AD} \big]\\[5pt]
  \textnormal{subject to}:\ & R_{AB}, Q_{AD} \geq 0, \, \\
  & Q_{AD}\otimes 1_B \geq R_{AB} \otimes 1_D \,.
  \end{array} 
\end{align}

Finally, we follow the argument leading Watrous to Alberti's expression for the fidelity~\cite[Lecture 8]{watrous-ln11}. We first remark that 
\begin{align}
\frac12 \tr \big[\rho_{AB} R_{AB}^{-1} \big] + \frac12 \tr \big[\sigma_{AD} Q_{AD} \big]\geq \sqrt{ \tr \big[\rho_{AB} R_{AB}^{-1} \big] \cdot \tr \big[\sigma_{AD} Q_{AD} \big]}
\end{align}
by the arithmetic--geometric mean inequality, with equality when the two terms are equal. However, it is easy to see that for any feasible pair $(R_{AB}, Q_{AD})$, there exists a constant $\lambda \in \mathbb{R}$ such that two trace terms evaluated for $(\lambda R_{AB}, \lambda Q_{AD})$ are equal (and clearly $(\lambda R_{AB}, \lambda Q_{AD})$ is also feasible). Hence, restricting our optimization to such rescaled pairs of operators, and going from the root fidelity to the fidelity again we find that
\begin{align}
F(\rho_{AB}\| \sigma_{AC})= \min_{\scriptstyle R_{AB} > 0, Q_{AD} > 0 \atop \scriptstyle Q_{AD}\otimes 1_B \geq R_{AB} \otimes 1_D } \tr \big[\rho_{AB} R_{AB}^{-1} \big] \cdot \tr \big[\sigma_{AD} Q_{AD} \big]\,.  \label{eq:alberti}
\end{align}
This concludes the proof.
\end{proof}


\section{Fidelity of Recovery is Multiplicative}\label{sec:mult}

As a direct consequence of this formulation of the problem we see that the fidelity of recovery is multiplicative.

\begin{prop}\label{lem:multi}
  For any $\rho_{AB} \in \cS(AB)$, $\tau_{A'B'} \in \cS(A'B')$, $\sigma_{AC} \in \cS(AC)$ and $\omega_{A'C'} \in \cS(A'C')$, we have
  \begin{align}
      F( \rho_{AB} \otimes \tau_{A'B'} \| \sigma_{AC} \otimes \omega_{A'C'})= F( \rho_{AB} \| \sigma_{AC}) \cdot F( \tau_{A'B'} \| \omega_{A'C'})\,.
  \end{align}
\end{prop}

\begin{proof}
From the definition in~\eqref{eq:def2} it is evident that if we restrict to recovery maps that have a product structure, we immediately find
\begin{align}
F( \rho_{AB} \otimes \tau_{A'B'} \| \sigma_{AC} \otimes \omega_{A'C'})\geq F( \rho_{AB} \| \sigma_{AC}) \cdot F( \tau_{A'B'} \| \omega_{A'C'}) \,.
\end{align}
To establish the equality, we take a closer look at~\eqref{eq:alberti}. Here we simply note the following. For every two pairs of feasible operators $(R_{AB}, Q_{AD})$ and $(R_{A'B'}, Q_{A'D'})$ for $F( \rho_{AB} \| \sigma_{AC})$ and $F( \tau_{A'B'} \| \omega_{A'C'})$, respectively, we have
\begin{align}
  Q_{AD}\otimes 1_B \geq R_{AB} \otimes 1_D \ &\land \ Q_{A'D'}\otimes 1_{B'} \geq R_{A'B'} \otimes 1_{D'} \nonumber\\
  &\implies\nonumber\\
  Q_{AD} \otimes Q_{A'D'} \otimes 1_{BB'}&\geq R_{AB} \otimes R_{A'B'} \otimes 1_{DD'} \,. \label{eq:some-people-dont-see-this}
\end{align}
To establish~\eqref{eq:some-people-dont-see-this} we used twice that $A \geq B \implies M \otimes A \geq M \otimes B$ for $M \geq 0$, which holds since taking the tensor product with $M$ is a positive map. Hence, $\big( R_{AB} \otimes R_{A'B'}, Q_{AD} \otimes Q_{A'D'} \big)$ is a feasible pair for $F( \rho_{AB} \otimes \tau_{A'B'} \| \sigma_{AC} \otimes \omega_{A'C'})$, and, thus, 
\begin{align}
F( \rho_{AB} \otimes \tau_{A'B'} \| \sigma_{AC} \otimes \omega_{A'C'})\leq \tr \big[\rho_{AB} R_{AB}^{-1} \big]\cdot\tr \big[\tau_{A'B'} R_{A'B'}^{-1} \big]\cdot \tr \big[\sigma_{AD} Q_{AD} \big] \cdot \tr \big[\omega_{A'D'} Q_{A'D'} \big] \,.
\end{align}
Since this holds for all feasible operators, we conclude that
\begin{align}
F( \rho_{AB} \otimes \tau_{A'B'} \| \sigma_{AC} \otimes \omega_{A'C'})\leq F( \rho_{AB} \| \sigma_{AC}) \cdot F( \tau_{A'B'} \| \omega_{A'C'}) \,.
\end{align}
\end{proof}


\section{The Fawzi-Renner Bound without de Finetti Reductions}\label{sec:fawzirenner}

Brand{\~a}o {\it et al.}~\cite{brandao14} show that the Fawzi-Renner lower bound~\eqref{eq:fawzirenner} can be deduced from the operational interpretation of the CQMI as (twice) the quantum communication cost in quantum state redistribution~\cite{devetak08,devetak09}. We simplify their proof and in particular get rid of the continuity and representation theoretic arguments (de Finetti reductions). Instead we leverage on the multiplicativity of the fidelity of recovery (Proposition~\ref{lem:multi}). We first give a precise definition for the information task of quantum state redistribution~\cite{devetak08,devetak09} (also see Figure~\ref{fig:redistribution} for an intuitive description).

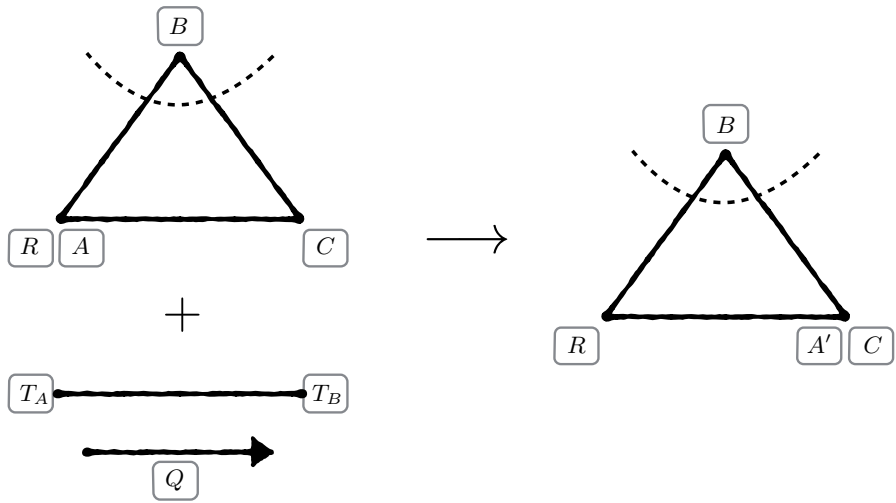
\begin{figure}[ht]
  \begin{overpic}[width=0.7\textwidth]{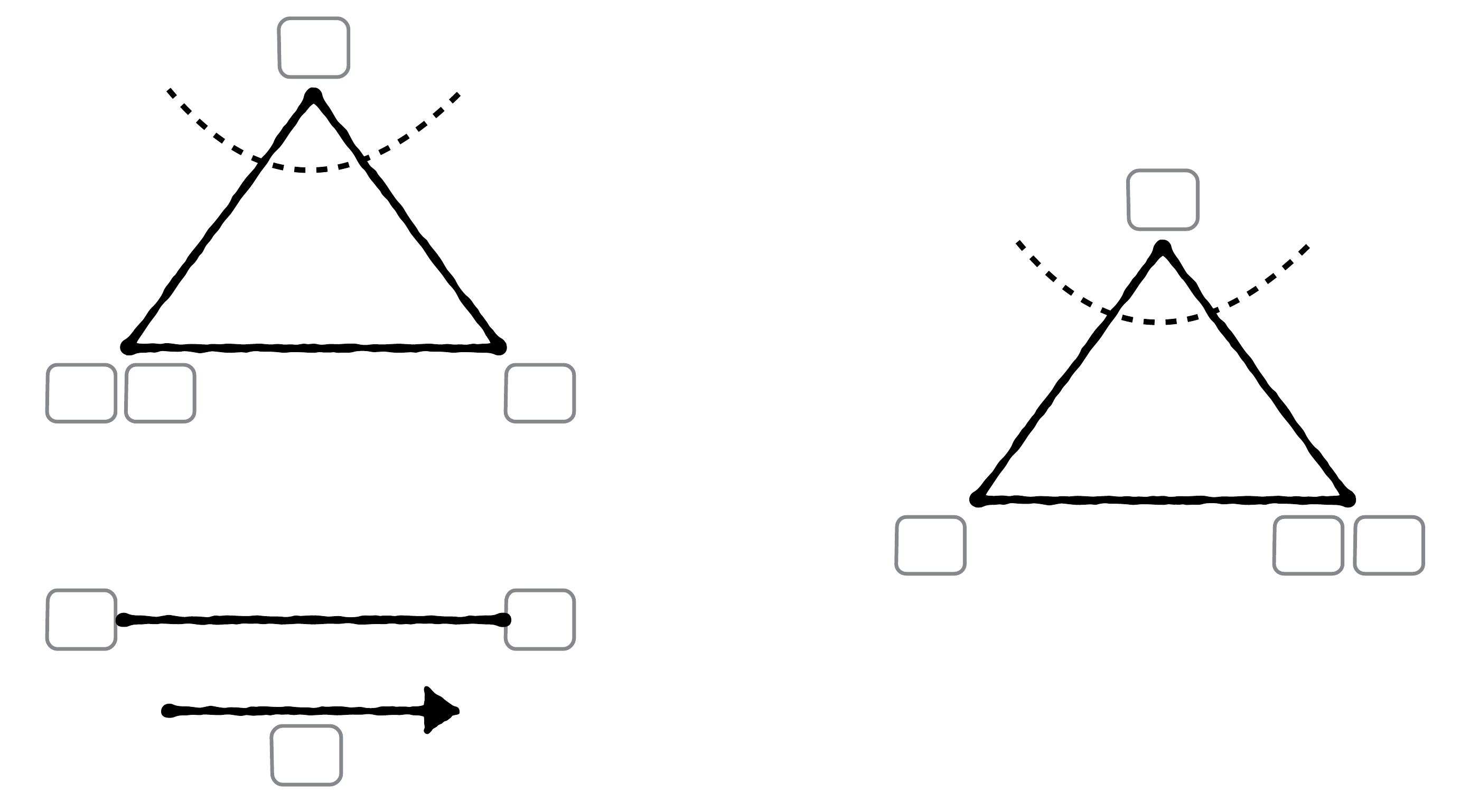}
    \put(19.3,51.2){ $B$}
    \put(3.5,27.9){ $R$}
    \put(9,27.9){ $A$}
    \put(34.9,27.9){ $C$}

    \put(18.9,3.3){ $Q$}
    \put(3.5,12.3){ $T_A$}
    \put(34.5,12.3){ $T_B$}
    
    \put(77.3,40.7){ $B$}
    \put(61.5,17.4){ $R$}
    \put(86.9,17.4){ $A'$}
    \put(92.9,17.4){ $C$}

    \put(19.5,20){\huge $+$}
    \put(47,28){\huge $\longrightarrow$}
  \end{overpic}
  \caption{Quantum state redistribution scheme. Starting from a four party pure state $\rho_{ABCR}$ and arbitrary entanglement assistance $T_AT_B$, the objective of quantum state redistribution is to transfer $A$ without altering the joint state. (See Definition~\ref{def:sm}). This is achieved by means of a local encoding operation $\mathcal{E}_{ART_A\to QR}$, transferring the system $Q$, and a local decoding operation $\mathcal{D}_{CT_{B}Q\to A'C}$. One is then interested in the trade-off between the fidelity of the scheme versus the quantum communication cost $q=\log d_Q$.}
  \label{fig:redistribution}
\end{figure}

\begin{definition}\label{def:sm}
Let $\rho_{ABCR} \in \cS(ABCR)$ be pure and let $QT_AT_B$ be additional spaces. A quantum state redistribution protocol for $A$ with fidelity $1-\delta$ is a pair of CPTP maps $\left(\cE_{ART_{A}\to QR},\mathcal{D}_{CQT_{B}\to A'C}\right)$, the encoder and the decoder, together with $\Phi_{T_{A}T_{B}}\in \cS(T_AT_B)$ such that
\begin{align}
F\Big(\left(\mathcal{I}_{BR}\otimes\mathcal{D}_{CQT_{B}\to A'C}\right)\circ\left(\cE_{ART_{A}\to QR}\otimes\mathcal{I}_{BCT_{B}}\right)\left(\rho_{ABCR}\otimes\Phi_{T_{A}T_{B}}\right),\rho_{A'BCR}\Big)\geq1-\delta\,,
\end{align}
where $\rho_{A'BCR}:=\left(\mathcal{I}_{A\to A'}\otimes\mathcal{I}_{BCR}\right)\rho_{ABCR}$. The number $q:=\log d_Q$ is called quantum communication cost of the protocol.
\end{definition}

This one-shot version of quantum state redistribution was analyzed in~\cite{bertatouchette14,datta14b,jain14}.

\begin{lemma}\cite[Theorem 4]{bertatouchette14}\label{lem:redistribution}
Let $\rho_{ABCR} \in \cS(ABCR)$ pure and $\eps>0$ sufficiently small. Then, there exists a quantum state redistribution protocol for $A$ with fidelity $1-144\eps^2$ and quantum communication cost
\begin{align}
q\leq\,\frac12 \Big( H_{\max}^{\eps}(A|C)_{\rho}-H_{\min}^{\eps}(A|BC)_{\rho} \Big)+\log(4/\eps^4)\,,
\end{align}
where $H_{\min}^{\eps}$ and $H_{\max}^{\eps}$ denote the smooth conditional min- and max-entropy, respectively.
\end{lemma}

This follows straightforwardly by reading~\cite[Theorem 4]{bertatouchette14} in terms of the fidelity (instead of the purified distance) and choosing appropriate error parameters. The precise definition of the smooth conditional min- and max-entropy and a discussion of their properties can be found in~\cite{tomamichel09}. However, here we will only need that they both asymptotically converge to the conditional von Neumann entropy. That is, the fully quantum asymptotic equipartition property~\cite{tomamichel08} tells us that for any $\eps\in(0,1)$,\footnote{These inequalities are in fact equalities~\cite[Section 6.4]{mybook}.}
\begin{align}\label{eq:aep}
\lim_{n\to\infty}\frac{1}{n}H_{\max}^{\eps}(A|C)_{\rho^{\otimes n}} \leq H(AC)_\rho-H(C)_\rho\quad\text{as well as}\quad
\lim_{n\to\infty}\frac{1}{n}H_{\min}^{\eps}(A|BC)_{\rho^{\otimes n}}\geq H(ABC)_\rho-H(BC)_\rho\,.
\end{align}
Hence, Lemma~\ref{lem:redistribution} implies the existence of a sequence (indexed by $n \in \mathbb{N}$) of quantum state redistribution protocols for $\rho_{ABC}^{\otimes n}$ with fidelity $\delta_n=1-144\eps_n^2$ converging to one for $n\to\infty$ and an asymptotic quantum communication cost rate $q_n$ satisfying
\begin{align}
\lim_{n\to\infty}\frac{q_n}{n}&\leq\lim_{n\to\infty}\left\{\frac{1}{n}\cdot\frac{1}{2}\Big(H_{\max}^{\eps_n}(A|C)_{\rho^{\otimes n}}-H_{\min}^{\eps_n}(A|BC)_{\rho^{\otimes n}}\Big)+\frac{1}{n}\log(4/\eps_n^4)\right\}\\
&\leq\frac{1}{2}\Big(H(AC)_\rho-H(C)_\rho -H(ABC)_\rho+H(BC)_\rho\Big)\\
&=\frac{1}{2}I(A\!:\!B|C)_{\rho}\,.
\end{align}
Following the ideas of Brand{\~a}o {\it et al.}~\cite{brandao14}, the existence of this protocol with a quantum communication cost rate given by (one-half) the CQMI can than be used to lower bound the CQMI.

\begin{theorem}\label{thm:main}
Let $\rho_{ABC} \in \cS(ABC)$. Then, we have
\begin{align}
I(A\!:\!B|C)_{\rho}\geq-\log F(A;B|C)_{\rho}\,.
\end{align}
\end{theorem}

\begin{proof}
The first part of the proof is the same as in~\cite{brandao14}. Let $\rho_{ABCR}$ be a purification of $\rho_{ABC}$ and consider any quantum state redistribution protocol for $A$ with encoder $\cE_{ART_{A}\to QR}$, decoder $\mathcal{D}_{CQT_{B}\to A'C}$, and define
\begin{align}
\sigma_{BCQRT_{B}}:=\left(\cE_{ART_{A}\to QR}\otimes\mathcal{I}_{BCT_{B}}\right)\left(\rho_{ABCR}\otimes\Phi_{T_{A}T_{B}}\right).
\end{align}
The authors of~\cite{brandao14} observe that the existence of a protocol with high fidelity is sufficient to give some bounds on the performance of the decoder even if the quantum communication is omitted. This is a consequence of the standard operator inequality $d_X\cdot1_{X}\otimes\omega_{Y}\geq\omega_{XY}$. In our case,
\begin{align}
\sigma_{BCQT_{B}} &\leq d_Q \cdot \sigma_{BCT_B}\otimes 1_Q = \rho_{BC}\otimes\left(d_Q^2\cdot\frac{1_Q}{d_{Q}}\right)\otimes\frac{1_{T_{B}}}{d_{T_{B}}}\,,
\end{align}
which again implies
\begin{align}
d_Q^2\cdot\mathcal{D}_{CQT_{B}\to AC}\left(\rho_{BC}\otimes\frac{1_Q}{d_{Q}}\otimes\frac{1_{T_{B}}}{d_{T_{B}}}\right)\geq\mathcal{D}_{CQT_{B}\to A'C}(\sigma_{BCQT_{B}})\,.
\end{align}
Now the following second part of the proof is different from~\cite{brandao14}. Due to the operator monotonicity of the square root function we get
\begin{align}\label{eq:fid-bound}
d_Q^2\cdot F\left(\rho_{A'BC},\mathcal{D}_{CQT_{B}\to A'C}\left(\rho_{BC}\otimes\frac{1_Q}{d_{Q}}\otimes\frac{1_{T_{B}}}{d_{T_{B}}}\right)\right)\geq F\left(\rho_{A'BC},\mathcal{D}_{CQT_{B}\to A'C}(\sigma_{BCQT_{B}})\right)\,.
\end{align}
By the protocol for quantum state redistribution (Lemma~\ref{lem:redistribution}), we find that for all $\eps>0$ small enough there exists a protocol such that
\begin{align}\label{eq:fid-bound2}
F\left(\rho_{A'BC},\mathcal{D}_{CQT_{B}\to A'C}(\sigma_{BCQT_{B}})\right)\geq1-144\eps^2
\end{align}
for a quantum communication cost
\begin{align}
\log d_Q \leq &\,\frac12 \Big( H_{\max}^{\eps}(A|C)_{\rho}-H_{\min}^{\eps}(A|BC)_{\rho} \Big)+\log(4/\eps^4)\,.
\end{align}
Hence, we can use~\eqref{eq:fid-bound} and~\eqref{eq:fid-bound2} to estimate
\begin{align}
H_{\max}^{\eps}(A|C)_{\rho}-H_{\min}^{\eps}(A|BC)_{\rho}&\geq -\log F\left(\rho_{A'BC},\mathcal{D}_{CQT_{B}\to A'C}\left(\rho_{BC}\otimes\frac{1_Q}{d_{Q}}\otimes\frac{1_{T_{B}}}{d_{T_{B}}}\right)\right)-\log(16/\eps^8)+\log(1-144\eps^2)\\
&\geq-\log\max_{\Gamma_{C\to A'C}}F\left(\rho_{A'BC},\Gamma_{C\to A'C}(\rho_{BC})\right)-\log(16/\eps^8)+\log(1-144\eps^2)\,,
\end{align}
and find that
\begin{align}
H_{\max}^{\eps}(A|C)_{\rho}-H_{\min}^{\eps}(A|BC)_{\rho}\geq -\log F(A;B|C)_{\rho} - O(\log(1/\eps))\,.
\end{align}
By applying this bound to $\rho_{ABC}^{\otimes n}$, multiplying the resulting inequality by $1/n$, and letting $n\to\infty$ we find by the asymptotic equipartition property for the smooth conditional min- and max-entropy~\eqref{eq:aep}, and the multiplicativity of the fidelity of recovery (Proposition~\ref{lem:multi}), that
\begin{align}
I(A\!:\!B|C)_{\rho}\geq-\log F(A;B|C)_{\rho}\,.
\end{align}
\end{proof}


\section{Conclusion}

In this paper we have generalized the FoR and have shown that the resulting measure is multiplicative for product states (Proposition~\ref{lem:multi}). It would be interesting to explore the consequences of this more for analyzing quantum correlations~\cite{seshadreesan14}, or for possible applications in computer science~\cite{braverman12,brandao13b,brandao13c,touchette14}. From the multiplicativity we also deduced an information theoretic proof of the Fawzi-Renner lower bound on the conditional quantum mutual information without making use of de Finetti reductions (Theorem~\ref{thm:main}). We note that Brand{\~a}o {\it et al.}~\cite{brandao14} also show a potentially stronger lower bound in terms of a regularized relative entropy distance,
\begin{align}
I(A\!:\!B|C)\geq \lim_{n\to\infty} \frac{1}{n} \min_{\Gamma_{C^n \to A^nC^n}} D\big(\rho_{ABC}^{\otimes n} \big\| \Gamma_{C^n \to A^nC^n} (\rho_{BC}^{\otimes n} ) \big) \,.
\end{align}
They then use de Finetti reductions to get rid of the regularization and arrive at a bound in terms of a measured relative entropy distance. On the other hand, our Proposition~\ref{lem:multi} rephrased in terms of the (sandwiched) quantum R\'enyi divergence of order $\frac12$~\cite{lennert13,wilde13} reads
\begin{align}
\min_{\Gamma_{C^n \to A^nC^n}} D_{\frac12} \big(\rho_{ABC}^{\otimes n} \big\| \Gamma_{C^n \to A^nC^n} (\rho_{BC}^{\otimes n} ) \big)= n\cdot\min_{\Gamma_{C \to AC}} D_{\frac12} (\rho_{ABC} \| \Gamma_{C \to AC} (\rho_{BC}) ) \,.
\end{align}
And hence the Fawzi-Renner bound then follows from the monotonicity of the quantum R\'enyi divergence in $\alpha$. Moreover, if such an additivity property would hold for any $\alpha \in (\frac12, 1]$ we could find stronger bounds. In particular, Li and Winter~\cite{liwinter14} have asked about a bound in terms of the relative entropy distance ($\alpha=1$). The corresponding problems can then no longer be phrased as SDPs but become complex optimization programs (for which duality is also available in principle~\cite{boyd04}).\footnote{For studying additivity properties one could also try to adapt the techniques of~\cite{hayashitomamichel14}.} Another interesting question is if we can estimate the performance of the optimal map $\Lambda_{C\to AC}$ in~\eqref{eq:fidrec} with the Petz recovery map~\cite{ohya93},
\begin{align}\label{eq:petz_map}
\Gamma_{C\to AC}^{\mathrm{Petz}}(\cdot):=\rho_{AC}^{1/2}\rho_{C}^{-1/2}(\cdot)\rho_{C}^{-1/2}\rho_{AC}^{1/2}\,,
\end{align}
in the sense that the Petz map should perform nearly as good as the optimal map. This in analogy to what Barnum and Knill have shown for the special case of pure states, i.e.~for the quantity~\eqref{eq:min_pure}. Finally, the results of Fawzi and Renner were recently generalized to the multiparty setting~\cite{wilde14} as well as to a new lower bound on the monotonicity of the quantum relative entropy under CPTP maps~\cite{berta14b}. It should be insightful to study SDP techniques in these more general settings as well. We also note that Piani~\cite{piani15} recently established a family of lower bounds on quantum discord using~\eqref{eq:fawzirenner}--\eqref{eq:fidrec} and SDPs similar to our primal problem in~\eqref{eq:primal}.

\textbf{Note added:} After completion of this work, improved lower bounds on the CQMI and on the monotonicity of the quantum relative entropy under CPTP maps were proven~\cite{renner15,wilde15,sutter15,junge15}. Moreover, using complex optimization duality the main result of this paper was extended to an additivity result for more general relative entropies of recovery~\cite{berta15}. It was also shown that the Petz recovery map~\eqref{eq:petz_map} does not perform square-root optimal compared to the optimal recovery map~\cite[Appendix F]{renner15} (even though some type of near optimality might still hold). Lastly it was pointed out very recently that there is a quantum interactive proof system whose maximum acceptance probability is equal to the fidelity of recovery~\cite{cooney15}, and hence the multiplicativity thereof also follows by a result of Kitaev and Watrous~\cite{kitaev00}.


\paragraph*{Acknowledgments.} We thank Fernando Brand{\~a}o, Omar Fawzi, Volkher Scholz, David Sutter, and Mark Wilde for discussions and feedback. MB acknowledges funding provided by the Institute for Quantum Information and Matter, an NSF Physics Frontiers Center (NFS Grant PHY-1125565) with support of the Gordon and Betty Moore Foundation (GBMF-12500028). Additional funding support was provided by the ARO grant for Research on Quantum Algorithms at the IQIM (W911NF-12-1-0521). MT is funded by an University of Sydney Postdoctoral Fellowship and acknowledges support from the ARC Centre of Excellence for Engineered Quantum Systems (EQUS). MB thanks the University of Sydney for hosting him while part of this work was done.


\bibliographystyle{arxiv_no_month}
\bibliography{library}

\end{document}